\documentclass[letterpaper, 10 pt, conference]{ieeeconf}

\usepackage{amsmath}

\let\labelindent\relax

\usepackage[font=small,skip=5pt]{caption}

\usepackage{tikz, pgfplots, pgf}
\usepackage{tikz}
\usetikzlibrary{calc, intersections, positioning}
\usetikzlibrary{positioning}
\usetikzlibrary{shapes.geometric, arrows}
\tikzstyle{bblock} = [rectangle, rounded corners, minimum width=0.5in, minimum height = 0.2in, text centered, draw=red!70, fill=red!5]
\tikzstyle{Bblock} = [rectangle, rounded corners, minimum width=1.65in, minimum height = 0.35in, text centered, draw=red!80, dashed]
\tikzstyle{arrow} = [thick, ->, >=stealth]
\tikzstyle{line} = [thick, -, >=stealth]
\tikzset{node distance = 0.9in and 0.2in}

\usepackage{amsthm}
\usepackage{amssymb}  
\usepackage{mathtools}
\usepackage[utf8]{inputenc}
\usepackage{comment}
\usepackage{algorithm,algpseudocode}
\algblock{Input}{EndInput}
\algnotext{EndInput}
\algblock{Output}{EndOutput}
\algnotext{EndOutput}

\usepackage{breqn}
\usepackage{graphicx}
\usepackage[version=4]{mhchem}
\usepackage{siunitx}
\usepackage{longtable,tabularx}
\usepackage{graphics} % for pdf, bitmapped graphics files
\usepackage{float}
\usepackage{epsfig} 
\usepackage{times} 
\usepackage{color}
\usepackage{pifont}
\usepackage{enumitem}
\usepackage{bm}
\usepackage{bbm}
\usepackage{soul}
\usepackage{balance}
\makeatletter
\newcommand{\multiline}[1]{%
  \begin{tabularx}{\dimexpr\linewidth-\ALG@thistlm}[t]{@{}X@{}}
    #1
  \end{tabularx}
}
\makeatother
\usepackage{hyperref}
\hypersetup{
    colorlinks=true,
    linkcolor=blue,
    filecolor=blue,      
    urlcolor=blue,
    pdftitle={Overleaf Example},
    pdfpagemode=FullScreen,
    }

\newcommand{\trace}{\textit{tr}}

\newcommand{\E}{\mathbb{E\,}}
\newcommand{\Cov}{\textbf{Cov}}

\newcommand{\R}{\mathbb{R}}

\newtheorem{assumption}{Assumption}

\newtheorem{lemma}{Lemma}
\newtheorem{remark}{Remark}

\newtheorem{proposition}{Proposition}
\newtheorem{corollary}{Corollary}

\newcommand{\add}[1]{{\color{black} #1}} % replace blue with black to add for real

\IEEEoverridecommandlockouts     
\newcommand{\MPar}[1]{\marginpar{}}

\usepackage{xargs}
\usepackage[colorinlistoftodos,prependcaption,textsize=tiny]{todonotes}
\newcommandx{\Mihai}[2][1=]{\todo[linecolor=red,backgroundcolor=red!25,bordercolor=red,#1]{#2}}

\title{\LARGE \bf
Extended Kalman filter--Koopman operator for tractable stochastic optimal control}

\author{Mohammad S. Ramadan,\,Mihai Anitescu
\thanks{The authors are with the Mathematics and Computer Science Division, Argonne National Laboratory, Lemont, IL 60439, USA,  {\tt\footnotesize mramadan@anl.gov, anitescu@mcs.anl.gov.}}}

\begin{document}
\maketitle
\thispagestyle{empty}
\pagestyle{empty}

%%%%%%%%%%%%%%%%%%%%%%%%%%%%%%%%%%%%%%%%%%%%%%%%%%%%%%%%%%%%%%%%%%%%%%%%%%%%%%%%
\begin{abstract}
The theory of dual control was introduced more than seven decades ago. Although it has provided rich insights to the fields of control, estimation, and system identification, dual control is generally computationally prohibitive. In recent years, however, the use of Koopman operator theory for control applications has been emerging. This paper presents a new reformulation of the stochastic optimal control problem that, employing the Koopman operator, yields a standard LQR problem with the dual control as its solution. We provide a numerical example that demonstrates the effectiveness of the proposed approach compared with certainty equivalence control, when applied to systems with varying observability.
\end{abstract}

\section{Introduction}

\add{Deterministic\MPar{1.1} control theory carries the implicit assumption of complete access to the states, a condition not met in many applications \cite{aastrom2012introduction}. Stochastic optimal control (SOC) \cite{feldbaum1960dual}, on the other hand, is hindered in practice by its computational complexity, limiting its benefits for the most part to the conceptual level via the introduction of dual control, its solution. Dual control has two key qualitative properties: caution and probing \cite{tse1973wide}. Caution accounts for uncertainty when achieving safety and improving performance, and it is manifested mostly through constraint tightening in tube-based model predictive control (MPC), or satisfying constraint instances in scenario-based methods \cite{ramadan2023control}. Probing, often in conflict with caution, reflects the online experiment design and active information-gathering roles of dual control to regulate the system's uncertainty.}

%Dual control is absent in stochastic linear systems with quadratic costs, as the state error covariance is not a function of the control input. For linear systems with unknown parameters, optimal probing can be manifested through the injection of a random signal \cite{marafioti2014persistently} that persistently excites the system. In general nonlinear systems, however, random noise injection is often insufficient. Instead, effective probing requires reasoning about the future measurement program \cite{tse1973wide}, leading to a computationally prohibitive formulation.

\add{We target solving SOC of a general differentiable nonlinear system and with a quadratic cost. We choose the extended Kalman filter (eKF) as the approximator of the state uncertainty propagation, since (i) the eKF is simple and widespread in navigation \cite{kendoul2012survey}, robotics \cite{jiang2017new}, power systems \cite{ghahremani2011dynamic}, and many other fields, rendering our developments in this paper immediately beneficial to numerous applications; (ii) the finite-dimensional approximation of the state uncertainty offered by the eKF (mean vector and covariance matrix) is appealing from the computational perspective; and (iii) the quadratic cost function can be written in terms of the state of the eKF.}

\add{The\MPar{1.2} intractability of SOC persists even when employing the eKF approximation of uncertainty \cite{tse1973wide}, thus ruling out conventional solution algorithms. Analogous to deterministic optimal control, solution algorithms to SOC have roots in: Dynamic Programming (DP) \cite{bertsekas2012dynamic}, which suffers from the curse of dimensionality, and its extensions to solve SOC \cite{Bayard2008Implicit} are generally cumbersome to implement and suffer from scalability limitations; or Pontryagin's principle, resembled mostly in nonlinear MPC \cite{rawlings2017model} over the stochastic dynamics (uncertainty dynamics), which requires high-dimensional nonconvex programming procedures that can be ill-suited for online computation. Other suboptimal SOC approaches exist \cite{mesbah2018stochastic}, and typically fall within these two principles, or restricted to limited formulations, such as linear systems with Gaussian parameters as in \cite{Heirung2017Dual}.}

\add{The\MPar{1.3} aforementioned nonlinear approaches have in common that their formulations are straightforward, while the challenge is in the solution. In contrast, our approach, which leverages the Koopman operator, switches the challenge from solution finding to problem formulation. That is, upon successful identification of a suitable basis dictionary for the Koopman operator to linearize the entire uncertainty propagation, the solution of SOC is a straightforward linear quadratic regulator (LQR) problem. This discussion of where the challenge / art is located, in the formulation or solution, resembles the discussion in \cite[p.~10]{boyd2004convex}, explaining the dichotomy between convex and nonconvex programs.}

\add{Other\MPar{1.4} works, such as \cite{surana2016linear,netto2018robust}, applied the Koopman operator for data-driven Kalman filtering and observer design. Our approach differs from these works in (i) context: ours is non-autonomous, i.e., for control design, hence, the Koopman operator is a parametric map of the control; and (ii) methodology: we solve the SOC problem by linearizing the whole state of the filter, including its uncertainty propagation, which is carried out by the eKF's Riccati equation.}

\add{Different\MPar{1.5} from the typical derivation of SOC \cite{tse1973wide,mesbah2018stochastic}, which starts from the stochastic DP equation, we use the smoothing theorem \cite{resnick2019probability} in a causality-respecting fashion. We then reformulate the SOC problem so that it is amenable to the application of Koopman operator theory. We discuss the structure of the new formulation and list the assumptions required for the continuity, bilinearity, and robust stability of the resulting new dynamics. The paper concludes with a numerical example with dynamics of varying observability and signal-to-noise ratio.}

\section{Problem Formulation} \label{section: ProblemFormulation} 
Consider the dynamic system
\begin{subequations}\label{eq:stateSpace}
\begin{align}
x_{k+1}&=f(x_k, u_k) + w_k, \label{eq:stateDynamics} \\
y_k &= h(x_k,u_k) + v_k, \label{eq:outputDynamics}
\end{align}
\end{subequations}
where $x_k\in\mathbb R^{r_x}$ is the state, $u_k\in\mathbb R^{r_u}$ the control input, and $y_k \in \mathbb R^{r_y}$ the measured output. The functions $f$ and $h$ are differentiable almost everywhere in their first arguments. The exogenous disturbances $w_k\in\mathbb R^{r_x}$, $v_k\in\mathbb R^{r_y}$ are each independent and identically distributed according to a density function. They are independent from each other and from $x_0$, have zero means and have covariances $\Sigma_w$ and $\Sigma_v$, respectively. The initial state $x_0$, prior to any measurement (including $y_0$), has a bounded mean $x_{0 \mid -1}$, a bounded covariance $\Cov(x_0) = \Sigma_{0 \mid -1}$, and a density $p_0$.

\begin{assumption} \label{Assumption1}
    (i) $\Sigma_v,\,\Sigma_w$, and $\Sigma_{0 \mid -1}$ are $\succ 0$ (positive definite) and bounded. (ii)\footnote{This point is typically required in system identification and the Koopman operator approximation methods and can be implied by ``nominal robust global asymptotic stability'' or ``positive invariance with disturbances'' as defined in \cite[p.~710]{rawlings2017model}.} $w_k \in W \subset \R^{r_x}$, $u_k \in \mathbb U\subset \R^{r_u}$, such that $W$ and $\mathbb U$ yield a compact set $\mathbb X$ invariant under \eqref{eq:stateDynamics}, starting from $x_0 \in \mathbb X$. (iii) The sequence of events is as follows: at time $k$, $u_k$ is applied; then $y_k$ becomes available.
\end{assumption}

The goal is to design a feedback control policy that minimizes the following quadratic cost $J^N(X_0,U_{N-1})$ (arguments suppressed for compactness):
\begin{align}
    J^N= \frac{1}{N} \cdot \E \Bigg \{ x_N^\top Q x_N + \sum_{k=0}^{N-1} x_k^\top Q x_k + u_k ^\top R u_k \Bigg \}, \label{eq:costRaw}
\end{align}
where $U_{N-1}=\{u_0,\hdots,u_{N-1}\}$, the tuple $X_0=( x_{0 \mid -1},\Sigma_{0 \mid -1})$, and $Q\succeq0$ (positive semi-definite) and $R\succ 0$. The expectation is over the probability space $\mathbb P_0$, characterizing the random variables $(x_0, V_{N-1}, W_{N-1})$, where $V_{N-1} = \{v_0,\hdots,v_{N-1}\}$ and $W_{N-1}=\{w_0,\hdots,w_{N-1}\}$, and the corresponding product Borel $\sigma$-field.

The control input is admissible when it is a causal law. It is restricted to be a function of the accessible data up to the time step of evaluating this law. That is, according to the $3$rd point of Assumption~\ref{Assumption1}, $u_k = u_k (Z_{k-1})$, where $Z_{k-1}=\{p_0,Y_{k-1},U_{k-1}\}$ stores all the past and accessible information up to time $k$, prior to implementing $u_k$. Here $Y_{k-1} = \{y_0,\hdots,y_{k-1}\}$. For consistency, we denote the prior information by $Z_{-1}=\{p_0\}$.

%%%%%%%%%%%%%%%%%%%%%%%%%%%%%%%%%%%%%%%%%%%%%%%%
%%%%%%%%%%%%%%%%%%%%%%%%%%%%%%%%%%%%%%%%%%%%%%%%
%%%%%%%%%%%%%%%%%%%%%%%%%%%%%%%%%%%%%%%%%%%%%%%%
%%%%%%%%%%%%%%%%%%%%%%%%%%%%%%%%%%%%%%%%%%%%%%%%
%%%%%%%%%%%%%%%%%%%%%%%%%%%%%%%%%%%%%%%%%%%%%%%%
\section{Methodology}
We first show that the cost \eqref{eq:costRaw} can be rewritten in terms of the first two moments of $x_k$.
\subsection{Equivalent description to the cost}
\begin{lemma} \label{lemma:smoothing_thm}
The term $\E x_k^\top Q x_k$ can be expressed by \MPar{3.3}
\begin{align*}
    \add{\E \left \{ x_k^\top Q x_k \right \} = \E \left \{  \trace(Q \Sigma_{k \mid k-1}) +  x_{k \mid k-1}^\top Q x_{k \mid k-1} \right \}.}
\end{align*}
The expectation to the left can be expressed in its integral form as $\E\{\cdot \} = \int \cdot\, p(x_k) dx_k$ with respect to the density function $p(x_k)$ \footnote{By the Markov property, $p(x_k)=p(x_0)p(x_1\mid x_0)\hdots p(x_k,x_{k-1})$, the initial density $p(x_0)=p_0(x_0)$ is given.}, while the one to the right can be expressed as $\E \left \{ \cdot \right \} = \int \cdot\, p(Y_{k-1}) dY_{k-1}.$ The first two conditional moments are
\begin{equation}
\begin{aligned}\label{eq:filterMeanCov}
     x_{k \mid k-1} &= \E \left \{ x_k \mid Z_{k-1}\right \},\\
    \Sigma_{k \mid k-1} &= \E \left \{ [x_k-x_{k \mid k-1}][x_k-x_{k \mid k-1}]^\top \mid Z_{k-1} \right \}.
\end{aligned}
\end{equation}
(The notation $()_{k \mid j}$ strictly follows that in \cite[Ch.~3]{anderson2012optimal}.)
\end{lemma}
\begin{proof}
This lemma is a direct consequence of the law of total expectation (smoothing theorem) \cite[p.~348]{resnick2019probability}, if we condition each additive term in \eqref{eq:costRaw} on its corresponding in-time $Z_{k-1}$ (respecting causality). That is,
\begin{align}
    &\hskip -5mm\E \left \{ x_k^\top Q x_k \right \} = \int x_k^\top Q x_k p(x_k)dx_k,\nonumber\\
    &= \int\int x_k^\top Q x_k p(x_k,Y_{k-1})d\,Y_{k-1} dx_k,\nonumber\\
    &= \int \left (\int x_k^\top Q x_k p(x_k \mid Y_{k-1}) dx_k \right )p(Y_{k-1}) dY_{k-1},\nonumber\\
    &=\E \left \{ \E \left \{ x_k^\top Q x_k  \mid Z_{k-1}\right\}\right \}.\label{eq:smoothing_theorem}
\end{align}

The state $x_k$ can be decomposed into $x_k =  x_{k \mid k-1} + \tilde x_k$, where $\tilde x_k$ is the estimation error. The quadratic term $x_k^\top Q x_k = ( x_{k \mid k-1} + \tilde x_k)^\top Q (x_{k \mid k-1} + \tilde x_k)$, under the conditional expectation $\E \left \{ \cdot \mid Z_{k-1} \right \}$ and, after ignoring the zero mean cross-terms, is equivalent to $ x_{k \mid k-1}^\top Q  x_{k \mid k-1} + \E \left \{ \tilde x_k Q \tilde x_k^\top \mid Z_{k-1} \right \}$. Using the cyclic property of the trace and the linearity of the expectation operator, we have $\E \left \{ \tilde x_k Q \tilde x_k^\top \mid Z_{k-1} \right \} = \trace( Q \Sigma_{k \mid k-1})$.
\end{proof}

\begin{proposition}
The cost function \eqref{eq:costRaw} can be represented by 
\begin{align}
     &J^N = \frac{1}{N} \cdot \E \Bigg [ x_{N\mid N-1}^\top Q  x_{N\mid N-1} + \trace(Q \Sigma_{N \mid N-1})  \nonumber\\
    &+ \sum_{k=0}^{N-1} \left [  x_{k\mid k-1}^\top Q  x_{k\mid k-1} + u_k ^\top R u_k + \trace(Q \Sigma_{k \mid k-1}) \right ] \Bigg ]. \label{eq:costWideSense}
\end{align}
\end{proposition}
\begin{proof}
Using Lemma~\ref{lemma:smoothing_thm} on each additive term in \eqref{eq:costRaw}, we get
\begin{align*}
    J^N &= \frac{1}{N} \cdot \E \Big \{ \E \left \{x_N^\top Q x_N \mid Z_{N-1} \right \} \nonumber \\
    & \hskip +10mm+ \sum_{k=0}^{N-1} \E \left \{x_k^\top Q x_k + u_k ^\top R u_k\mid Z_{k-1} \right \} \Big \}.
\end{align*}
Then we use the derivation subsequent to \eqref{eq:smoothing_theorem} for each conditional expectation above.
\end{proof}

\subsection{Evolution of the central moments}
In addition to \eqref{eq:filterMeanCov}, let
\begin{align*}
     x_{k \mid k} &= \E \left \{ x_k \mid Z_{k-1},y_k \right \}, \\
    \Sigma_{k \mid k} &= \E \left \{ [x_k-x_{k \mid k}][x_k-x_{k \mid k}]^\top \mid Z_{k-1},y_k \right \}.
\end{align*} 
We use the eKF to propagate the central moments $x_{k \mid k-1}, \Sigma_{k \mid k-1}$ appearing in \eqref{eq:costWideSense}. At each time step, the eKF is the recursion
\begin{align}
 x_{k+1 \mid k} = f( x_{k \mid k},u_k),\quad \label{eq:eKF_state}
\Sigma_{k+1 \mid k} = F_k \Sigma_{k \mid k} F_k^\top + \Sigma_w, 
\end{align}
where
\begin{equation}
\begin{aligned} \label{eq:eKF_supportVariables}
& x_{k\mid k}= x_{k \mid k-1} + \Omega_{k} \left [y_{k}-h(x_{k \mid k-1}, u_k)\right],\\
&\Sigma_{k \mid k} = \left [I - \Omega_{k} H_{k} \right]\Sigma_{k \mid k-1},\\
 &\Omega_{k}= \Sigma_{k \mid k-1} H_{k}^\top \left [ H_{k} \Sigma_{k \mid k-1} H_{k}^\top+\Sigma_v\right]^{-1}, \\
 &F_k = \left. \frac{\partial f(x,u_k)}{\partial x} \right | _{x_{k\mid k}},\quad H_{k} = \left. \frac{\partial h(x,u_k)}{\partial x} \right | _{x_{k\mid k-1}}.
\end{aligned}
\end{equation}
\add{The\MPar{3.1} eKF extends the Kalman filter to nonlinear systems via employing a first-order approximation formed by the Jacobians $F_k$ and $H_k$, with $\Omega_k$ mimicking the Kalman gain \cite{anderson2012optimal,jazwinski2007stochastic}.} The recursion is initialized by $ x_{0 \mid -1}$, $\Sigma_{0 \mid -1}$, given in Section~\ref{section: ProblemFormulation}.

We note that the cost description in \eqref{eq:costWideSense} is exact, because the cost is quadratic. However, the central moments $ x_{k \mid k-1}$ and $\Sigma_{k \mid k-1}$ provided by the eKF above are only approximates and not exact. The reason is that, in general, $f$ and $h$ are nonlinear and the disturbances $w_k,v_k$ are not necessarily Gaussian \cite{anderson2012optimal}. 

\begin{assumption} \label{Assumption2}
    (i) The eKF estimation error $\lVert x_k - x_{k \mid k-1} \rVert_2$ is bounded, and (ii) the covariance $\Sigma_{k \mid k-1}$ \eqref{eq:eKF_state} is positive definite and bounded, for all $k$\MPar{3.0}.\footnote{The first point is not straightforward to guarantee mathematically \cite{la1995conditions} but may be justified by the estimation accuracy and the stability of the eKF in various applications. Further discussion of this condition can be found in \cite[Sec.~9.6]{jazwinski2007stochastic}. The covariance boundedness can be achieved by satisfying the uniform observability condition in \cite{reif1999stochastic}. \add{The second point guarantees that the factorization is well defined without pivoting/permutations.}}
\end{assumption}

The variables $x_{k \mid k-1},\Sigma_{k \mid k-1}$ are random since they are functions of the random observation sequence $Y_{k-1}$ (not yet available at $k=0$). The expectation in \eqref{eq:costWideSense} averages over $Y_{k-1}$. Next, we employ an important approximation to omit this expectation.

\subsection{Certainty equivalence of the information state}

In the LQG context, the separation principle \cite{aastrom2012introduction} states that optimal control can be separated into deterministic optimal control and optimal filtering or, equivalently, into LQR and LQE. However, this does not hold for general nonlinear systems. Hence, to omit the expectation in \eqref{eq:costWideSense}, we require some assumptions.

\begin{assumption} \label{Assumption3}
The measurement correction term $[y_{k} - h(x_{k\mid k-1}, u_k)]$, in \eqref{eq:eKF_supportVariables}, is independent from $\Omega_{k}$ and is a zero-mean white noise sequence.\footnote{
This term bears a passing resemblance to the innovation sequence of the Kalman filter. If the system \eqref{eq:stateSpace} is linear, the eKF reduces to the Kalman filter, and this sequence is therefore white and of zero mean \cite[Sec.~5.3]{anderson2012optimal}. In the general nonlinear case, the complete whiteness, zero mean, and independence conditions are not guaranteed but are satisfied to some extent in various applications (see \cite[Sec.~8.2]{anderson2012optimal}).
}$^,$.
\end{assumption}

This removes the need for the expectation in \eqref{eq:costWideSense} and motivates a measurement-free version of the eKF in \eqref{eq:eKF_state}. Let
\begin{align} \label{eq:CE_step}
     x^p_{k+1} =f( x^p_{k },u_k).
\end{align}
The state $x^p_{k}$ is the surrogate of both $ x_{k\mid k-1}$ and $x_{k \mid k}$ (since the measurement correction is omitted) that evolves solely through prediction. 

\begin{remark} \label{remark1}
The surrogate state $x_k^p$, defined by employing certainty equivalence in the SOC sense (CE-SOC), is deterministic. It will be used offline for learning the Koopman operator and the design of a control law. However, the original eKF state $ x_{k\mid k-1}$ will be used when this control is implemented online in closed-loop.
\end{remark}

We define the CE information state dynamics as
\begin{align}
     \pi_{k+1}
    =  T_{\pi} \left ( \pi_k
    , u_k\right ), \label{eq:system_eKF_CE}
\end{align}
where the tuple $ \pi_k = (x^p_{k}, \Sigma^p_{k \mid k-1})$, and $\Sigma^p_{k \mid k-1},\, \Omega^p_k,\, F^p_k,\, H^p_k$, are defined as their corresponding in-name variables but with $x^p_{k}$ replacing both $x_{k \mid k}$ and $x_{k \mid k-1}$, that is, with $[y_k-h(x_{k \mid k-1}, u_k)]$ set to zero. The initial conditions $ x^p_0 = x_{0 \mid -1}$ and $ \Sigma^p_{0 \mid -1} = \Sigma_{0 \mid -1}$, are given in Section~\ref{section: ProblemFormulation}.

Corresponding to \eqref{eq:costWideSense}, the CE-SOC cost is
\begin{align}
      &\hat J^N(X_0,U_{N-1}) = \frac{1}{N} \Bigg [ x_{N}^{p\top} Q x^p_{N} + \trace(Q \Sigma^p_{N \mid N-1}) + \nonumber\\
    & \sum_{k=0}^{N-1} \left [  x_{k}^{p\top} Q x^p_{k} + u_k ^\top R u_k + \trace(Q \Sigma^p_{k \mid k-1}) \right ] \Bigg ]. \label{eq:costWideSenseDeterministic} 
\end{align}
Note that all the involved variables in \eqref{eq:costWideSenseDeterministic} are deterministic, avoiding the high-dimensional integration over the many stochastic variables as in \eqref{eq:costWideSense}.

The following subsections are geared toward reformulating the dynamic system \eqref{eq:system_eKF_CE} into a linear realization and then using the Koopman operator theory \cite{korda2018linear} to put into the LQR context the computation of a feedback control law that minimizes \eqref{eq:costWideSenseDeterministic}.

\subsection{Structure of $T_\pi$}
Note that in \eqref{eq:system_eKF_CE}, $\Sigma^p_{k+1 \mid k}$ is a function of $ x^p_{k}$ and $u_k$, through $H^p_{k}$ and $F^p_k$. This makes $\Sigma^p_{k+1\mid k}$ a nonlinear function of $u_k$ (it is nonlinear in $H^p_{k}$) unless $u_k \mapsto f( x^p_k,u_k)$ is affine in $u_k$ with constant coefficients, for any $x^p_k \in \mathbb X$. To be precise, we call a function $f_0$ defined on $\mathbb X \times \mathbb U$ \textit{bilinear} if it can be written as $f_0(x,u) = f_1(x) + f_2(x) u$, $f_1$ and $f_2$ of the appropriate dimensions, and we call it \textit{bilinear with constant coefficients} if $f_2(x)=B_f$, a constant matrix.

We note that even if $f$ and $h$ are affine in $u_k$, not necessarily with constant coefficients, $\Sigma^p_{k+1 \mid k}$ is in general a nonlinear function of $u_k$, since $u_k$ still appears in the Jacobians $F^p_k$ and $H^p_k$.

\vskip 3mm
\begin{proposition}
If $f$ and $h$ in \eqref{eq:stateSpace} are bilinear in $u_k$ with constant coefficients, then: (i) the function $T_\pi$ is bilinear in $u_k$ with constant coefficients, (ii) $T_\pi$ is continuous almost everywhere in the elements of the tuple $\pi_k$, and (iii) the covariance $\Sigma^p_{k+1 \mid k}$ is a function of $x^p_{k}$ only and is independent from $u_k$.\qed
 \label{prop:bilinearity}
\end{proposition}

This proposition imposes more structure on the dynamics \eqref{eq:system_eKF_CE}, enabling the application of more specialized control algorithms requiring bilinearity in the input, and allowing the discussion of truncation error, convergence, and stability under different model realizations \cite{bruder2021advantages,iacob2024koopman}. The third point is important in that it denies the immediate effect of $u_k$ on the state covariance evolution; that is, $\Sigma^p_{k+1 \mid k}$ is independent of $u_k$. This in turn allows various transformations (for example,  Cholesky factorization) of the covariance matrix, without complicating the effect of $u_k$ on the resulting transformed dynamics. 

\subsection{Toward the standard LQR form}
Let $L_k$ be the Cholesky (lower triangular) factor of $\Sigma^p_{k \mid k-1} = L_k L_k^\top$. Let $\ell_k$ be the half-vectorized (the nonzero elements, column by column, from left to right) of $L_k$. We denote the (invertible) mapping from the tuple $\pi_k = ( x^p_{k},\Sigma^p_{k \mid k-1})$ to the vector $\eta_k = [ x_{k}^{p\top}, \ell_{k}^\top]^\top$,
\begin{align} \label{eq:pi2eta}
\eta_k = \mathcal{M}(\pi_k).
\end{align}
We also denote the evolution of $\eta_k$ by
\begin{align*}
    \eta_{k + 1} = T_\eta(\eta_k, u_k).
\end{align*}
The following describe the well-definedness, the structure of the above reformulation, and their relationship to $T_{\pi}$.

\vskip 3mm
\begin{lemma} \label{lemma:continuity_of_T_3}
The function $\mathcal{M}$ in \eqref{eq:pi2eta} and its inverse $\mathcal{M}^{-1}$ are well defined and continuous. Moreover, $\eta_{k + 1} = T_\eta(\eta_k, u_k) = \mathcal{M}\left (T_\pi(\mathcal{M}^{-1}(\eta_k),u_k) \right)$, and $T_\eta$ is continuous almost everywhere in $\eta_k$, for any $u_k$.
\end{lemma}
\begin{proof}
In the $x^p_k$ portion, $\mathcal{M}$ is simply an identity map. 

From Assumptions~\ref{Assumption1} and \ref{Assumption2}, the matrix $\Sigma^p_{k \mid k-1}$ is positive definite \cite{reif1999stochastic}, which in turn implies that the Cholesky decomposition is continuous \cite[p.~295]{schatzman2002numerical}. The inverse of the decomposition is simply obtained by $\Sigma^p_{k \mid k-1} = L_k^\top L_k$, which is also continuous. The half-vectorizing map $L_k \mapsto \ell_k$ and its inverse are obviously continuous. The continuity $T_\eta$ follows from the above and $T_\pi$ being continuous almost everywhere in $\eta_k$.
\end{proof}

\begin{corollary} \label{corollary:blinearity_of_T_3}
For all $k$, $\eta_k \in \mathbb H \subset \R^{r_\eta}$, $\mathbb H$ is compact.
\end{corollary}
\begin{proof}
It follows from Lemma~\ref{lemma:continuity_of_T_3} that the image under a continuous function of a compact set is compact, and $\mathbb X \ni x_k$ according to Assumption~\ref{Assumption1} is compact.
\end{proof}

Now we show that the above transformation turns the cost \eqref{eq:costWideSenseDeterministic} into the standard LQR form.
\vskip 3mm
\begin{proposition}
The cost in \eqref{eq:costWideSenseDeterministic} is equivalent to \begin{align}
     \hat J^N= \frac{1}{N} \Bigg [\eta_{N}^\top Q_{\star \star}  \eta_{N} + \sum_{k=0}^{N-1} \left [ \eta_{k}^\top Q_{\star \star} \eta_{k} + u_{k} ^\top R u_{k} \right ]\Bigg],\label{eq:costConcat}
\end{align}
where $Q_{\star \star} = \text{block-diag}(Q, Q_{\star})$, and $Q_{\star} = \text{block-diag}(Q_1, Q_2, \hdots, Q_{r_x})$, where $Q_i$ is the principal submatrix of $Q$ starting from the element $(i,i)$ of $Q$ to the end (the element $(r_x,r_x)$, for example, $Q_1=Q$)). 
\end{proposition}
\begin{proof}
The term $\trace(Q \Sigma^p_{k \mid k-1})$ in \eqref{eq:costWideSenseDeterministic}, using the cyclic property of the trace, can be written as $\trace(L_k^\top Q L_k)$, which equivalently can be described by $\trace(L_k^\top Q L_k) = \ell_k ^\top  Q_{\star} \ell_k$. By substituting $\eta_k$ in place of $x^p_k$ and $\Sigma^p_{k \mid k-1}$, \eqref{eq:costWideSenseDeterministic} turns into \eqref{eq:costConcat}.
\end{proof}

\subsection{Koopman and eDMD for control}
We follow the notation and the generalization of the Koopman operator to systems with control inputs as presented in \cite{korda2018linear}. We let $l_u$ be the space of all control sequences of the form $\{u_k\}_{k=0}^\infty=:\mathbf{u}$, and $\mathbb H \times l_u$, where $\mathbb H \subset \R^{r_\eta}$ is a compact set. The ``uncontrolled dynamics'' are
\begin{align} \label{eq:eta_u_dynamics}
    \begin{bmatrix}
    \eta_{k+1}\\
    \mathbf{u}
    \end{bmatrix}
    =:
    \chi_{k+1}
    =
    \begin{bmatrix}
    T_\eta(\eta_k, \mathbf{u}(k))\\
    q \mathbf{u}
    \end{bmatrix}
    =:
    \mathcal{T}(\chi_k),
\end{align}
where $q$ is the time-shift operator applied elementwise: $q \mathbf{u}(k) = \mathbf{u}(k+1)$, and $\mathbf{u}(k) := u_k$. 

The Koopman operator corresponding to \eqref{eq:eta_u_dynamics}, assuming its forward invariance, $\mathcal{K}: \mathcal C(\mathbb H \times l_u) \to \mathcal C(\mathbb H \times l_u)$ ($\mathcal C(\star)$ is the space of real-valued continuous functions with domain $\star$), is defined by $\mathcal{K} \phi (\chi_k) = \phi \circ \mathcal{T} (\chi_k),$ for every real-valued function $\phi$ with domain $\mathbb H \times l_u$.

Let $\Psi$ be a finite-dimensional vector of elements of $\mathcal C(\mathbb H)$, such that $\Psi = \left [\psi_1,\psi_2,\hdots,\psi_{N_{\Psi}} \right ].$ This vector is solely a function of the state $\eta_k$. We let $\Psi_k := \Psi(\eta_k)$ and construct $\Psi$ such that the first $r_\eta$ elements of $\Psi_k$ are $\eta_k$ and such that it spans a subspace of $\mathcal C(\mathbb H)$.\footnote{For technical consideration, we choose $\psi(\eta_k,a) = \psi(\eta_k,b)$ for all $a,b \in l_u$ and suppress the second argument.}

\begin{figure}
\centering
\begin{tikzpicture}[align=center, node distance= 0.93in and 0.2in] 
\scriptsize
\node (space0){$u_k$};
\node (dynamics) [bblock, right of = space0] {$T_{\pi}(\pi_k,u_k)$};
\node (pi2eta) [bblock, right of = dynamics] {$\mathcal{M}(\pi_k)$};
\node (eta2psi) [bblock, right of = pi2eta] {$\Psi(\eta_{k+1})$};
\node[label={ $T_{\eta}(\eta_k,u_k)$}] (space3) [Bblock] at ($(pi2eta)!0.57!(dynamics)$) {};
\node  (EDMD) [bblock, below of = space3, yshift=0.1in] {eDMD \\ (DMD)};
\node  (LQR) [bblock, right of = EDMD] {LQR \\ design};
\node (helper1x) at ($(space0)!0.7!(dynamics)$) {};
\node (helper2x) at ($(pi2eta)!0.3!(eta2psi)$) {};
\node (stackU) [below of = helper1x, yshift=0.5in] { $
\begin{bmatrix}
    \lvert &  & \lvert\\
    u_{{d}} & \hdots& u_0\\
    \lvert &  & \lvert
\end{bmatrix}$ \\ $ U_{data}$};
\node (stackPsi) [below of = helper2x, yshift=0.5in] { $
\begin{bmatrix}
    \lvert &  & \lvert\\
    \Psi_{d} & \hdots& \Psi_0\\
    \lvert &  & \lvert
\end{bmatrix}$ \\  $\Psi_{data}$};
\draw [arrow]  (stackU) |-  node[anchor=south] {} (EDMD);
\draw [arrow]  (EDMD) --  node[anchor=south] {} (LQR);
\draw [arrow]  (stackPsi) -|  node[anchor=south] {} (EDMD);
\draw [arrow]  (space0) --  node[anchor=south] {} (dynamics);
\draw [arrow] (dynamics) -- node[anchor=south] {$\pi_{k+1}$} (pi2eta);
\draw [arrow] (pi2eta) -- node[anchor=south] {$\eta_{k+1}$} (eta2psi);
\draw [arrow] (space0) |- node[anchor=north] { record} (stackU);
\draw [arrow] (eta2psi) |- node[anchor=north] { record} (stackPsi);
\end{tikzpicture}
\caption{\label{Fig:dataCollectionStage}\textbf{(Training Stage):} Block diagram illustrating the steps of the simulation-data collection stage, starting from a randomized initial condition $\eta_0$ and an injected, persistently exciting control sequence $\{u_k\}_{k \geq 0}$, next the creation of the data matrices $U_{data}$ and $\Psi_{data}$, then the application of eDMD, and ending with solving the LQR problem of the lifted system.}
\end{figure}
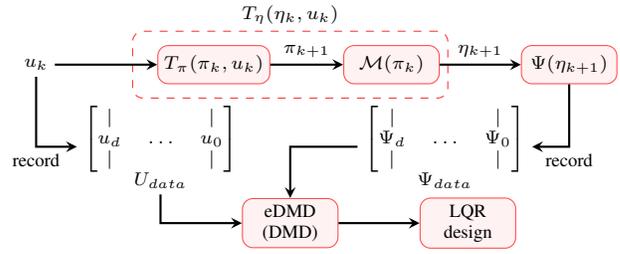

We now use the extended dynamic mode decomposition (eDMD) procedure, a simulation- or data-driven approach that seeks a finite-dimensional approximation of the Koopman operator \cite{korda2018linear}. After the data collection stage, as illustrated in Figure~\ref{Fig:dataCollectionStage}, the eDMD approximation of a linear representation of the Koopman operator can be found through solving the following optimization problem,
\begin{align} \label{eq:linearOptimization}
    \min_{ A,\, B} \quad \lVert \Psi_{data}^+ -  A \Psi_{data}^- -  B U_{data} \rVert_{F},
\end{align}
where $\lVert \cdot \rVert_{F}$ is the Frobenius norm, $\Psi_{data}^+ = [\Psi_{d}, \hdots, \Psi_1]$, and $\Psi_{data}^- = [\Psi_{d-1}, \hdots, \Psi_0]$. The unique solution to \eqref{eq:linearOptimization} under the full column rank condition of the data is $[ A,\, B] = \Psi_{data}^+    \begin{bmatrix}
        \Psi_{data}^-\\
        U_{data}
    \end{bmatrix}
    ^\dagger.
$ Since $\eta_k$ forms the first entries of $\Psi_k$, it can be recovered by the canonical projection $ C = [\mathbb I_{r_\eta \times N_\Psi}\,\, 0]$. The predictor is now given by 
\begin{align}
    \Psi_{k+1} =  A \Psi_k +  B u_k,\quad \eta_k =  C \Psi_k. \nonumber
\end{align}
(For the quality of this approximation see \cite{haseli2023invariance,iacob2024koopman}.)

The cost function \eqref{eq:costConcat} can now be described by
\begin{align}
\hat J^N = \frac{1}{N} \Bigg [\Psi_{N}^\top \mathcal{Q}  \Psi_{N} + \sum_{k=0}^{N-1} \left [ \Psi_{k}^\top \mathcal Q \Psi_{k} + u_{k} ^\top R u_{k} \right ]\Bigg], \label{cost_Psi}
\end{align}
where $\mathcal{Q} = \text{block-diag}(Q_{\star \star},\,0_{N_\Psi-r_\eta \times N_\Psi-r_\eta}) \succeq 0$. We assume $( A,\, B)$ to be stabilizable and $( A,\,\mathcal{Q}^{\frac{1}{2}})$ detectable (see  \cite{mamakoukas2023learning} for imposing these assumptions). Hence, $J^N \to a < \infty$ as $N \to \infty$, for a stabilizing feedback law. We denote the resulting LQR control (of the lifted-dynamics) by $K_{\Psi}$. Figure~\ref{Fig:closed-loopStage} illustrates the implementation phase of the control law $K_\Psi$ in closed-loop with the original system. The difference in the used variables between Figures~\ref{Fig:dataCollectionStage} and \ref{Fig:closed-loopStage} is explained in Remark~\ref{remark1}.

\begin{figure}
\centering
\begin{tikzpicture}[align=center, node distance = 0.93in and 0.2in] \scriptsize

\node (dynamics) [bblock] {state-space\\model \eqref{eq:stateSpace}\\ ``plant''};
\node (eKF) [bblock, right of = dynamics] {eKF \\\eqref{eq:eKF_state}};
\node (lag) [bblock, right of = eKF] {lag: $z^{-1}$};
\node (pi2eta) [bblock, below of = lag, yshift=0.5
in] {$\mathcal{M}(\pi_{k}')$};
\node (eta2Psi) [bblock, left of = pi2eta] {$\Psi(\eta_{k}')$};
\node (LQR_gain) [bblock, left of = eta2Psi] {$K_{\Psi} \times \cdot$};

\draw [arrow] (dynamics) -- node[anchor=south] {$y_k$} (eKF);
\draw [arrow] (eKF) -- node[anchor=south] {$\pi_{k+1}'$} (lag);
\draw [arrow] (lag) -- node[anchor=west] {$\pi_k'=$ \\$(x_{k\mid k-1},\Sigma_{k \mid k-1})$} (pi2eta);
\draw [arrow]  (pi2eta) --  node[anchor=south] {$\eta_{k'}$} (eta2Psi);
\draw [arrow]  (eta2Psi) --  node[anchor=south] {$\Psi_{k}'$} (LQR_gain);
\draw [arrow] (LQR_gain) --  node[anchor=east] {$u_{k}$} (dynamics);
\end{tikzpicture}
\caption{\label{Fig:closed-loopStage}\add{\textbf{(Implementation stage):} A block-diagram illustrating the online implementation of the SOC-LQR control of the lifted system in closed-loop over the original system ``plant'' \eqref{eq:stateSpace}. The variables with a prime are functions of $y_k$, since the CE step \eqref{eq:CE_step} is omitted in the implementation phase, according to Remark~\ref{remark1}.}}
\end{figure}
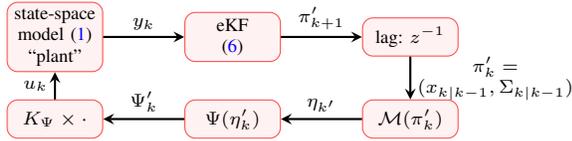

\section{Numerical Example} \label{Section: Numerical Examples}
In this section we implement our linearized SOC approach on a system with varying state observability over the state-space. Suppose we have the system
\begin{align}
x_{k+1} &= 
     \begin{bmatrix}
         .63 &.54 & 0\\
         .74 &.96 & .68\\
         .1 & -.86 & .54     \end{bmatrix}
     x_k + 
     \begin{bmatrix}
         0\\
         1\\
         0
     \end{bmatrix}
     u_k + w_k, \label{eq:stateDynamicsExample}
     \\
    y_k &= \text{\small ELU}\left (\sum_{i=1}^3x_k^i-3\right ) + v_k, \nonumber\label{eq:outputDynamicsExample}
\end{align}
where $x_k=(x_k^1,x_k^2,x_k^3)^\top$, $\text{\small ELU}$ is the exponential linear activation function, widely used in deep learning applications, which is $\text{\small ELU}(x)=x$ for $x\geq0$ and plateaus toward $\text{\small ELU}(x)=-1$ for $x \ll 0$ (in particular, it is $e^x-1,\,x<0$). The processes $w_k$ and $v_k$ follow the same assumptions as for \eqref{eq:stateSpace} in Section~\ref{section: ProblemFormulation}. \add{Furthermore, \MPar{4.1}$w_k \sim \mathcal{N}^3(0,0.2 \mathbb{I}_{3 \times 3})$ ($\mathcal{N}^{\text{trunc}}(\mu, \Sigma)$ is a Gaussian of zero mean and covariance $\Sigma$, truncated beyond the Mahalanobis distance $\text{MD}\geq \text{trunc}$ about the mean, then renormalized\footnote{\add{The corresponding Mahalanobis distance is given by $\text{MD}(x)=\sqrt{(x-\mu)^\top \Sigma^{-1} (x-\mu)}$. The set $\{x\in \R \mid \text{MD}(x)\leq 2\}$ corresponds to a $\approx 95\%$ confidence interval for a univariate normal. For a multivariate normal on $\R^r$, $\text{MD}^2(x)$ is the Chi-squared density with $r$ degrees of freedom, since it is a sum of squared $r$ independent and standardized ($\sqrt{\Sigma^{-1}}(x-\mu)$) Gaussian random variables. When $r=3$, the set $\{x \in \R^3 \mid \text{MD} \leq 3\}$ corresponds to a $\approx 97\%$ confidence region (ellipsoid centered at $\mu$). Truncation and renormalization are done implicitly via rejection sampling; that is,  sample $x$ is rejected if $\text{MD}(x) > \text{trunc}$, and sampling is repeated.}}$^,$\footnote{\add{The covariance after truncation can still be approximated by $\Sigma$ when the value trunc chosen represents a high percentage confidence region.}}), $v_k \sim \mathcal{N}^2(0,0.2)$ and $x_0 \sim \mathcal{N}^3(0,\mathbb{I}_{3 \times 3})$. These truncations are required to satisfy the boundedness in Assumption~\ref{Assumption1}.}

This model belongs to an emerging class of models in system identification: the Hammerstein-Wiener family, which consists of linear systems composed in series with algebraic nonlinearities \cite{wills2013identification}. \add{We\MPar{4.2} pick this example because: (i) the corresponding SOC problem, with $9$-dimensional state-space, whether through nonconvex MPC or DP, can be cumbersome or even prohibitive, and (ii) it splits the state space into two half-spaces separated by the surface $\sum x^i - 3 = 0$, with one half-space on which the system has significantly higher observability than the other. Hence, the behavior of the dual control derived is immediately interpretable.}
 
\add{Notice\MPar{4.3} that, since $\text{ELU}(x)=x$ for $x\geq 0$, if $\sum_{i=1}^3x_k^i \geq 3$, $y_k$ has sensitivity w.r.t. $x_k$, while if $\sum_{i=1}^3x_k^i \ll 3$, this sensitivity vanishes (compared with the constant $y_k$'s sensitivity to $v_k$, i.e., $\partial y_k / \partial v_k = 1$). This difference in sensitivity between the two half-spaces ($\mathcal H_1 = \{\sum_{i=1}^3x_k^i < 3\}$ and $\mathcal H_2 = \{ \sum_{i=1}^3x_k^i \geq 3\}$) affects the observability of $x_k$ in the complete stochastic observability sense defined in \cite{liu2011stochastic}. In particular, for any value $\sum_{i=1}^3x_k^i \ll 3$, $y_k \approx -1 + v_k$, $x_k$ cannot be identified.\footnote{Different definitions of observability for nonlinear systems exist. For example, the one provided by \cite{vaidya2007observability} (Def.~12: Degree of Observability)  is related to the magnitude (power) of $y_k$ itself rather than the ability to infer $x_k$ from it. For $\sum_{i=1}^3x_k^i \ll 3$, this definition returns a nonzero degree of observability, since $y_k \approx -1 + v_k \neq 0$ almost surely, contradicting the lack of identifiability in such regions of the state space. Therefore, the definition of the complete stochastic observability in \cite{liu2011stochastic} is more aligned with our intention. The two definitions align in the linear case where the power of $y_k$ is captured in the observability Gramian, which also implies the identifiability of the state.} Therefore, the quality of estimation varies over the state space. For better observability, it is required to ``kick'' the state to the half-space $\mathcal{H}_2$, opposing stabilization, which, instead, requires the state to stay close to zero, contained in the opposite half-space $\mathcal H_1$.}

In the training stage, and according to Figure~\ref{Fig:dataCollectionStage}, we inject a random white noise input $u_k \sim \mathcal{N}^2(0,0.2)$. We follow Figure~\ref{Fig:dataCollectionStage} for data collection and then use a simple linear DMD model of the state $\Psi_k=[\eta_k^\top,1]^\top$.

For the control design, we pick $Q=\mathbb I_{3 \times 3}$ ($\mathcal Q$ can be found accordingly) and $R = 1$, then find the LQR control law $K_{\Psi}$ in \eqref{cost_Psi} (as $N \to \infty$). We apply $K_\Psi$ in a closed-loop simulation as explained in Figure~\ref{Fig:closed-loopStage}. We also calculate the LQR control law $K$ of the deterministic version of the system \eqref{eq:stateDynamicsExample} ($w_k=0$ and $y_k=x_k$, i.e., the CE assumption). The comparison between the performance of these two controllers is shown in Figure~\ref{fig:closed_loop_sim} and Table~1, after implementing each control design (CE-LQR: $K x_{k \mid k-1}$ and SOC-LQR: $K_\Psi \Psi_k'$, with $\Psi_k'$ as in Figure~\ref{Fig:closed-loopStage}). SOC-LQR achieves a better control cost and a lower estimation error. That is, SOC-LQR not only  improves the control performance but also significantly increases the estimation quality (the eKF performance) $\epsilon := \sum_k \lVert x_{k \mid k-1} - x_k^{true} \rVert^2_2$ \footnote{These results can be reproduced by using our open-source \textsc{Julia} code found at \href{https://github.com/msramada/linearizing-uncertainty-for-control}{github.com/msramada/linearizing-uncertainty-for-control}.}. This can also be seen in Figure~\ref{fig:compared_to_true}, where the true state (known in simulation) is wandering around for CE-LQR, while it is estimated well and thus controlled well with SOC-LQR.

\begin{figure}
\centering 
\includegraphics[width=3.2in,height=1.7in]{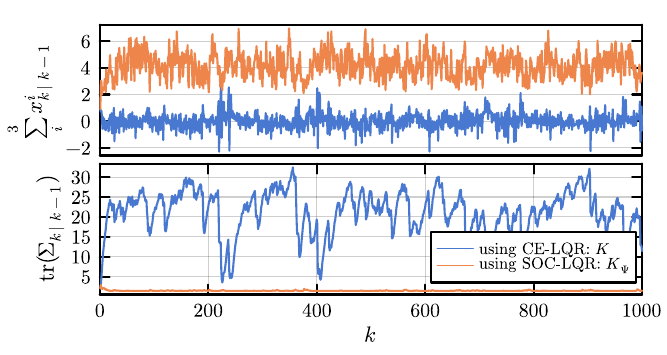} 
\caption{\add{The SOC-LQR pushes the state to $\mathcal{H}_2$ (top), resulting in significantly lower estimation covariance (bottom).}}\label{fig:closed_loop_sim}
\end{figure}

\begin{figure}
\centering 
\includegraphics[width=3.2in,height=0.85in]{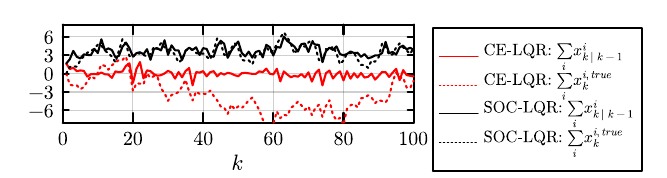} 
\caption{\add{SOC-LQR accurately controls and estimates the true state, whereas CE-LQR assumes the state is around zero, while estimation error is huge and true state is wandering.}}\label{fig:compared_to_true}
\end{figure}

\begin{table}[h]
\begin{center}
\footnotesize
\begin{tabular}{ |c||c|c|c| } 
 \hline
 \multicolumn{4}{|c|}{Table 1: Control performance and estimation quality: $K$ vs. $K_\Psi$.} \\
 \hline
 Metric (time-averaged)& CE-LQR: $K$ & SOC-LQR: $K_\Psi$ & reduction\\
 \hline
 Achieved cost   & $22$    &$1.7$ & $92\%$\\
 $\epsilon$ & $28.9$  & $0.52$ & $98\%$\\
 \hline
\end{tabular}
\end{center}
\end{table}
\section{Conclusion}
The proposed approach requires the eKF to be an accurate Bayesian filter for the problem at hand. For systems with more complex uncertainty descriptions, an alternative filter might be required. Our methodology is also contingent, in its linearization part, on the choice of the dictionary functions used and, hence, on the prediction accuracy of the computed approximate Koopman operator. Therefore, the future of SOC will benefit immediately from advancements in the Koopman operator theory for control.

We seek further investigations toward exploiting the algebraic structure of the eKF and incorporating deep embedding approaches \cite{tiwari2023computationally} to address large-scale SOC problems.

\section*{Acknowledgment}
This material was based upon work
supported by the U.S. Department of Energy, Office of Science,
Office of Advanced Scientific Computing Research (ASCR) under
Contract DE-AC02-06CH11347. 

\bibliographystyle{IEEEtran}
%\bibliography{IEEEabrv,mybibfile} %\bibliographystyle{ieeetr}        % Include this if you use bibtex 
\bibliography{References}

\vspace{0.1cm}
\begin{flushright}
	\scriptsize \framebox{\parbox{2.5in}{Government License: The
			submitted manuscript has been created by UChicago Argonne,
			LLC, Operator of Argonne National Laboratory (``Argonne").
			Argonne, a U.S. Department of Energy Office of Science
			laboratory, is operated under Contract
			No. DE-AC02-06CH11357.  The U.S. Government retains for
			itself, and others acting on its behalf, a paid-up
			nonexclusive, irrevocable worldwide license in said
			article to reproduce, prepare derivative works, distribute
			copies to the public, and perform publicly and display
			publicly, by or on behalf of the Government. The Department of Energy will provide public access to these results of federally sponsored research in accordance with the DOE Public Access Plan. http://energy.gov/downloads/doe-public-access-plan. }}
	\normalsize
\end{flushright}	

\end{document}